\documentclass{article}
\usepackage[utf8]{inputenc}
\usepackage{authblk}

\title{Reducing Transducer Equivalence to Register Automata Problems solved by ``Hilbert Method''}
\author[]{Adrien Boiret}
\author[]{Radosław Piórkowski}
\author[]{Janusz Schmude}
\affil[]{University of Warsaw}

\date{}

\usepackage{tikz}
\usetikzlibrary{trees}
\usetikzlibrary{calc}

\usepackage{amsmath}
\usepackage{amssymb}
\usepackage{amsthm}
\usepackage{stmaryrd}

\usepackage{amsfonts}
\usepackage{graphicx}
\usepackage{xspace}
\usepackage{todonotes}
\usepackage{comment}
\usepackage{enumitem}
\usepackage{hyperref}
\usepackage{bookmark}
\usepackage{nameref}
\usepackage{ifthen}
\usepackage{algorithm}
\usepackage[noend]{algpseudocode}
\usepackage{tikz}
\usepackage[normalem]{ulem}
\usepackage[all]{xy}

\usepackage{macros}
\usepackage{comment-color}

\newif\iflong

\longfalse

\begin{document}
	
	\maketitle			
	\begin{abstract}
		In the past decades, classical results from algebra, including Hilbert's Basis Theorem,
		had various applications in formal languages, including a proof of the Ehrenfeucht Conjecture, decidability of HDT0L sequence equivalence,
		and decidability of the equivalence problem for functional tree-to-string transducers.
		
		In this paper, we study the scope of the algebraic methods mentioned above,
		particularily as applied to the functionality problem for register automata, and equivalence for functional register automata.
		We provide two results, one positive, one negative.
		The positive result is that functionality and equivalence are decidable for MSO transformations on unordered forests.
		The negative result comes from a try to extend this method to decide functionality and equivalence on macro tree transducers.
		We reduce macro tree transducers equivalence to an equivalence problem
		for some class of register automata naturally relevant to our method.
		We then prove this latter problem to be undecidable.
	\end{abstract}

	\pagebreak
	\section{Introduction}\label{secIntroduction}
	
	The study of finite-state machines, such as transducers~\cite{Rounds70,Choffrut03,Maneth15} or register automata~\cite{AlurCerny10,AlurDAntoni17},
	and of logic specifications, such as MSO-definable transformations~\cite{Courcelle94},
	provides a theoretical ground to study document and data processing.
	
	In this paper,
	we will consider the equivalence problem of functional transducers.
	We focus on register automata, i.e. transducers that store  values in a finite number of registers that can be updated or combined
	after reading an input symbol.
	Streaming String Transducers (SST)~\cite{AlurCerny10} and Streaming Tree Transducers (STT)~\cite{AlurDAntoni17} are classes of register automata (see for example~\cite{AlurDDRY13}) where the equivalence is decidable
	for the \emph{copyless} restriction, i.e. the case where each register update cannot use the same register twice.
	This restriction makes SST equivalent to MSO-definable string transformations.
	Macro tree transducers (MTT)~\cite{EngelfrietVogler85}, an expressive class of tree transducers for which equivalence decidability remains a challenging open problem,
	can be seen as register automata, whose registers store tree contexts.
	Although equivalence is not known to be decidable for the whole class, there exists a linear size increase fragment of decidable equivalence,
	that is equivalent to MSO-definable tree transformations, and can be characterized by a restriction on MTT quite close to copyless~\cite{EngelfrietManeth03}.
	
	
	Some equivalence decidability results have been proven on register automata without copyless restrictions~\cite{SeidlMK15,BenediktDSW17},
	by reducing to algebraic problems 
	such as 
	ideal inclusion and by applying Hilbert's Basis Theorem and other classical results of algebraic geometry.
	In this paper we will refer to this as the ``Hilbert Method''. 
	This method was used to prove diverse results, dating back to at least the proof of the Ehrenfeucht Conjecture~\cite{AlbertLawrence85}, and the sequence problem for HDT0L~\cite{CulikKarhumaki86,Honkala00}.
	It has recently found new applications in formal languages; for example, equivalence was proven decidable for general tree-to-string transducers
	by seeing them as copyful register automata on words~\cite{SeidlMK15}.
	
	In this paper, we use an abstraction of these previous applications of the ``Hilbert Method'' as presented in~\cite{toolbox}.
	We apply these preexisting results to the study of unordered forest transductions -- and notably MSO functions.
	Note that equivalence of MSO-definable transductions on unordered forests is not a straightforward corollary of the ordered case, as the loss of order makes equivalence more difficult to identify.
	We 
	also try to apply those methods to obtain decidability of MTT equivalence.
	For unordered forests, we obtain a positive result, showing that register automata on forest contexts with one hole have decidable functionality and equivalence.
	For the attempt to study MTT, we prove an undecidability result on register automata using polynomials and composition,
	which means the natural extension of this approach does not yield a definitive answer for the decidability of MTT equivalence.
	
	\subsection*{Layout}
	
	\textbf{Section~\ref{secPrelim}} presents the notions of algebra, register automata, and
	the notions necessary to use an abstraction of the ``Hilbert Method'' as presented in~\cite{toolbox}.
	\textbf{Section~\ref{secForests}} is dedicated to the proof of the positive result that
	we can apply the ``Hilbert Method'' to contexts of unordered forests with at most one hole (i.e. the algebra of unordered forests with limited substitution).
	This provides a class of register automata encompassing MSO functions on unordered forests where functionality is decidable.
	Finally, \textbf{Section~\ref{secPolySub}} describes how 
	applying a method similar as in Section~\ref{secForests}
	to study MTT equivalence
	leads to studying register automata on the algebra of polynomials with the substitution operation,
	a class whose functionality and equivalence we prove to be undecidable.
	\section{Preliminaries}\label{secPrelim}
	
	\partitle{Algebras}
	An \emph{algebra} $\aalg=(\aalgset,\enum{\opalg}1n)$ is a (potentially infinite) set of elements $\aalgset$, and a finite number of operations $\enum{\opalg}1n$.
	Each operation is a function $\opalg:\aalgset^k\rightarrow \aalgset$ for some $k\in\Nat$.
	
	
	
	\partitle{Polynomials}
	For an algebra $\aalg$ and a set $X=\{\enum x1n\}$ of variables, we note $\aalgset[X]$ the set of terms over $A\cup X$.
	A \emph{polynomial function} of $\aalg$ is a function $f:\aalgset^k\rightarrow \aalgset$. 
	For example on $\aalg=(\Rat,+,\times)$, the term $\times(+(x,2),+(x,y))$ induces the polynomial function $f:(x,y)\mapsto (x+2)(x+y)$.
	The definition of polynomial functions can be extended to functions $f\colon\aalgset^k\rightarrow \aalgset^m$ by product of their output:
	if $\enum f1m$ are polynomial functions from $\aalgset^k$ to $\aalgset$,
	then $f'\colon \ol a\in\aalgset^k\mapsto (f_1(\ol a),\dots, f_m(\ol a))$ is a polynomial function from $\aalgset^k$ to $\aalgset^m$.
	Note that polynomial functions are closed under composition.
	
	One can define the \emph{algebra of polynomials over $\aalg$ with variable set $X$}, denoted $\aalg[X]$.
	Its elements are equivalence classes of terms over $\aalgset\cup X$ with the operations of \aalg, where two terms are called \emph{equivalent} if they induce identical polynomial functions. 
	$\aalg[X]$ can be seen as an algebra that subsumes \aalg, with natural definition of operations.
	A classical example of this construction is the ring of polynomials $(\Q[x],+,\times)$, obtained from the ring $(\Q,+,\times)$.
	

	
	By adding the substitution operation \pcomp-- to $\aalg[X]$,
	we get a new algebra called a \emph{composition algebra of polynomials} and denoted \aalgsubs.
	Homomorphisms of such algebras are called \emph{composition homomorphisms}.
	For brevity we write $\comp {-}{x_i}{-}$ for the substitution of a single $x_i\in X$.
	Examples of such algebras include well-nested words with a placeholder symbol ``$?$'', as used in the registers of Streaming Tree Transducers~\cite{AlurDAntoni17},
	or tree contexts with variables in their leaves, as used in Macro Tree Transducers~\cite{EngelfrietVogler85}.
	
	\partitle{Simulation}
	Following the abstractions as they are presented in~\cite{toolbox}
	\footnote{As of this version's redaction, Part 11 of~\cite{toolbox} is the part relevant for this paper. This and any theorem or page number can change in future versions of~\cite{toolbox}},
	we define simulations between algebras in a way that is relevant to the use of the ``Hilbert Method''.
	\begin{definition}\label{defSimulation}
		Let $\aalg$ and $\balg$ be algebras. We say that $\alpha : \aalgset \to \balgset^n$ is a \emph{simulation} of \aalg in \balg
		if for every operation $\opalg : \aalgset^m \to \aalgset$ of $\aalg$,
		there is a polynomial function $f : \balgset^{m \times n} \to \balgset^n$ of $\balg$
		such that $\alpha\circ \opalg=f\circ\alpha$, where $\alpha$ is defined from $\aalgset^m$ to $\balgset^{m \times n}$ coordinate-wise.
		If such a simulation $\alpha$ exists, we say that \aalg is \emph{simulated by} \balg ($\aalg\simulatedBy\balg$).
		\begin{figure*}[h]
			\begin{center}
				\begin{tikzpicture}
				\node (am) at (0,0) {${A^m}$};
				\node (a1) at (0,-2) {${A}$};
				\node (bmn) at (3,0) {$B^{m\times n}$};
				\node (bn) at (3,-2) {$B^n$};
				
				\draw[->] (am) edge node[left] {$\opalg$} (a1);
				\draw[->] (am) edge node[above] {$(\alpha,\cdots,\alpha)$} (bmn);
				\draw[->] (bmn) edge node[right] {$f$} (bn);
				\draw[->] (a1) edge node[below] {$\alpha$} (bn);
				\end{tikzpicture}
			\end{center}
		\end{figure*}
	\end{definition}
	The following lemma states that simulations extend to composition algebras.
	\begin{lemma}\label{lemSimSub}
		Let $\boldsymbol{Q}[X]=(\Q[X],+,\times)$.
		If $\aalg\simulatedBy \boldsymbol{Q}[X]$ and $\aalgset$ is an infinite set, then $\aalg[Y]^\submark\simulatedBy \boldsymbol{Q}[X][Y]^\submark$.
	\end{lemma}
	

	\iflong
\else
	\begin{proof}
		If there is a simulation $\alpha$ from $\aalg$ to $\Q[X]$,
		then $\alpha$ can be extended into a simulation $\widetilde{\alpha}$ from $\aalg[Y]$ to $\boldsymbol{Q}[X][Y]$
		by setting $\alpha(Y)=Y$, and requiring $\widetilde{\alpha}$ to be a homomorphism. 
		It is important to check that $\widetilde{\alpha}$ indeed is a function (i.e. preserves equivalence of terms): if terms $t_1, t_2\in\aalgset[Y]$ 
		induce the same functions on $A$,
		then $\widetilde{\alpha}(t_1)(Y)$ and $\widetilde{\alpha}(t_2)(Y)$ are polynomial functions that are equal on $\alpha(A)$. Since $\alpha(A)$ is an infinite subset of $\Q[X]$, %
		$\alpha(t_1)$ and $\alpha(t_2)$ are equal everywhere. 
		The proof of injectivity is straightforward. Let $P(Y), Q(Y)\in\aalgset[Y]$ be any two nonequivalent terms. Then there is a tuple $\ol a$ of $A$ such that $P(\ol a)\neq Q(\ol a)$.
		If $\widetilde{\alpha}(P)=\widetilde{\alpha}(Q)$, we would have $\alpha(P(\ol a))=\alpha(P)[Y:=\ol a]=\alpha(Q)[Y:=\ol a]=\alpha(Q(\ol a))$.
		This would contradict the injectivity of $\alpha$ on $\aalgset$.
	\end{proof}
	
	\fi
	
	\partitle{Typed Algebras} Some of the algebras we consider are \emph{multi-sorted},
	which is to say that their elements are divided between a finite number of types.
	A \emph{multi-sorted algebra} is an algebra $\aalg=(\aalgset,\enum{\opalg}1n)$ such that:
	\begin{itemize}
		\item $\aalgset$ can be partitioned into $\enum{\aalgset}1m$,
		\item each operation $\rho$ is a function $\opalg:\aalgset_{i_0}\times\dots\times\aalgset_{i_k}\rightarrow \aalgset_j$.
	\end{itemize}
	To each $a\in\aalgset$ we associate a \emph{type}, which is a unique $i$ such that $a\in\aalgset_i$.
	Note that in a multi-sorted algebra \aalg
	polynomial functions are typed $f:\aalgset_{i_0}\times\dots\times\aalgset_{i_k}\rightarrow \aalgset_j$,
	and simulation and substitutions must be defined type-wise.
	
	\partitle{Register automata}
	In this paper we will work on register automata that make a single bottom-up pass on an input ranked tree,
	use a finite set of states, and a finite set of registers with values in \aalgset for some algebra $\aalg$.
	When the automaton reads an input symbol, it updates its register values as a polynomial function of \aalg applied to the register values in its subtrees.
	This formalism is already present in the literature: streaming tree transducers~\cite{AlurCerny10}, for example,
	are register automata on input words and register values in the algebra of words on an alphabet $\SigIn$, with the concatenation operation.
	
	A \emph{signature} $\SigIn$ is a finite set of symbols $a$, each with a corresponding finite rank $\rank(a)\in\Nat$.
	A \emph{ranked tree} is a term on this signature $\SigIn$: if $a\in\SigIn$, $\rank(a)=n$, and $\enum t1n$ are trees, then $a(\enum t1n)$ is a tree.
	
	\begin{definition}\label{defRA}
		Let $\aalg=(\aalgset,\enum{\opalg}1n)$ be an algebra. A \emph{bottom-up register automaton} with values in \aalg (or \aRA\aalg)
		is a tuple $\M=(\SigIn,n,Q,\rul,\outfun)$, where:
		\begin{itemize}
			\item $\SigIn$ is a ranked set
			\item $n$ is the number of $\aalgset$-registers used by $\M$
			\item $Q$ is a finite set of states
			\item $\rul$ is a finite set of transitions of form $a(\enum q1k)\rightarrow q,f$
			where $a\in\SigIn$ of rank $k$, $\lbrace q,\enum q1k\rbrace\subset Q$, and $f:\aalgset^{n\times k}\rightarrow \aalgset^n$ a polynomial function of \aalg.
			\item $\outfun$ is a partial output function that to some states $q\in Q$ associates $f_q:\aalgset^n\rightarrow\aalgset$ a polynomial function of \aalg.
		\end{itemize}
		A \emph{configuration} of \M is a $n$-uple $(q,\ol r)$ where $q\in Q$ is a state and $\ol r=(\enum r1n)\in\aalgset^n$ is a $n$-uple of register values in \aalgset.
		We define by induction the fact that a tree $t$ can reach a configuration $(q,\ol r)$, noted $t\rightarrow(q,\ol r)$:
		If $a\in\SigIn$ of rank $k$, $a(\enum q1k)\rightarrow q,f$ a rule of \rul,
		and for $0\leqslant i\leqslant k$, $t_i\rightarrow(q_i,\ol r_i)$,
		then $$a(\enum t1k)\rightarrow (q,f(\enum{\ol r}1k)).$$
		
		\M determines a relation $\sem M$ from trees to values in \aalgset. It is defined using $\outfun$ as a final step:
		if $\outfun(q)=f_q$, and $t\rightarrow(q,\ol r)$, then $f_q(\ol r)\in\sem\M(t)$.
	\end{definition}
	
	We say that a \aRA\aalg is \emph{functional} if \sem\M is a function.
	We say that a \aRA\aalg is \emph{deterministic} if for all $a,\enum q1k$ there is at most one rule $a(\enum q1k)\rightarrow q,f$ in \rul.
	Any deterministic \aRA\aalg is functional.
	
	Note that on a multi-sorted algebra, we further impose that
	every state $q$ has a certain type $\aalgset_{i_1}\times\dots\times\aalgset_{i_n}$,
	i.e. if $(q,\ol r)$ is a configuration of $M$, then $\ol r \in\aalgset_{i_1}\times\dots\times\aalgset_{i_n}$.
	
	\partitle{``Hilbert Method''}
	We now describe an abstraction~\cite{toolbox} of the classical algebra methods
	that are used in the literature~\cite{SeidlMK15,BenediktDSW17}
	to decide equivalence of functional register automata over certain algebras (e.g. $(\Sigma^*,\cdot)$)
	using what we will refer to as the ``Hilbert Method''.
	More specifically, as it is always easy to prove the semi-decidability of non-equivalence of functional \aRA\aalg,
	by guessing two runs on the same input with different outputs,
	this method aims to prove that the functionality and equivalence problems over functional \aRA\aalg are semi-decidable.

	This method can be described as a 4-step process:
	
	\begin{itemize}
		\item Simulate \aalg by $\boldsymbol{Q}$, hence reducing \aRA\aalg equivalence to \aRA{$\boldsymbol{Q}$} equivalence
		\item Functional \aRA{$\boldsymbol{Q}$} equivalence can be reduced to \aRA{$\boldsymbol{Q}$} zeroness, i.e. checking if a \aRA{$\boldsymbol{Q}$} only outputs 0.
		\item \aRA{$\boldsymbol{Q}$} zeroness can be reduced to ideal inclusion problem in $\ol\Q[X]$, i.e.  the ring of polynomials with algebraic numbers as coefficients
		\item Ideal inclusion problem in $\ol\Q[X]$ is decidable
	\end{itemize}

	These results exist in the literature. We will provide references as well as an intuition of the main mechanisms in these proofs.
	For the first point, the reduction from \aRA\aalg equivalence to \aRA{$\boldsymbol{Q}$} equivalence, an example is provided in~\cite{SeidlMK15}, where \aalg is $(\Sigma^*,\cdot)$.
	In essence, this part of the method amounts to a simulation as described in Definition~\ref{defSimulation}.
	If $\aalg\simulatedBy\balg$ with simulation $\alpha$, then any \aRA{\aalg} $M$ can be simulated by a \aRA{\balg} $M'$,
	in the sense that $M$ outputs $a\in\aalgset$ for an input $t$ if and only $M'$ outputs $\alpha(a)\in\balgset^k$ for the same input $t$.
	This gives a reduction from \aRA\aalg equivalence to \aRA\balg equivalence.

	The second point is presented in the proof of Theorem 11.8 of~\cite{toolbox}.
	If $M$ and $M'$ are two functional \aRA{$\boldsymbol{Q}$} of same domain,
	using a natural product construction,
	one can create $M''$ that runs $M$ and $M'$ in parallel,
	then computes the difference of outputs between $M$ and $M'$.
	Thus $M$ and $M'$ are equivalent iff $M''$ only outputs 0.

	The third point can be found in the proof of Theorem 11.8 of~\cite{toolbox}.
	The idea is to express \aRA{$\boldsymbol{Q}$} zeroness as a set problem
	(with polynomial grammars as an intermediary in~\cite{toolbox}).
	We want to find for each state $q$ the set $X_q$ of register values that $M$ can hold in state $q$.
	These states obey to some inclusion equations: if $a(\enum q1k)\rightarrow q,f$ is a rule of $M$, then $f(X_{q_1}\times...\times X_{q_n})\subseteq X_q$.
	Furthermore, if zeroness is true for $M$, then for every $q$ such that final output function $f_q$ is defined, $f_q(X_q)\subseteq\lbrace0\rbrace$.
	Interestingly, if such a family of sets of $\Q$
	$(X_q)_{q\text{ state of }M}$ exists to satisfy those inclusions,
	then there exists a family of ideal sets of $\ol\Q$,
	$(S_q)_{q\text{ state of }M}$, that also satisfy those inclusions (Lemma 11.5 of~\cite{toolbox}).

	The fourth point uses classical algebra results to find such a family of ideal sets.
	The proof, as it is presented in Theorem 11.3 of~\cite{toolbox}, works as follows:
	Hilbert's Basis Theorem ensures that all families of ideals $(S_q)_{q\text{ state of }M}$ can be enumerated.
	For each of these families, it can be checked using Groebner Basis whether it respects a set of inclusions or not.
	Eventually, if a solution exists, it will be found, making this ideal problem, and thus \aRA{$\boldsymbol{Q}$} zeroness, semi-decidable.

	For this paper, we point out a few natural extensions to those methods
	and establish Theorem~\ref{thmPolyRA} and Corollary~\ref{corRA-Equivalence}
	as a basis for our work.
	
	The first remark is that one can consider more problems than functional equivalence. Functionality itself can be studied with these methods. It is preserved by simulations, and \aRA{$\boldsymbol{Q}$} functionality can be reduced to zeroness: instead of comparing two functional
	\aRA{$\boldsymbol{Q}$} in the second point, $M''$ can run two copies of the same \aRA{$\boldsymbol{Q}$} $M$ and compute the output difference.
	$M$ is functional iff $M''$ only outputs 0.
	
	The second remark is that the classical algebra results (Hilbert Basis Theorem, Groebner Basis, algebraic closure of a field...)
	used in the fourth point extend to any computable field $\mathcal{K}$.
	In consequence,
	Theorem~11.8 of \cite{toolbox} holds for any \aRA{$\mathcal{K}$}.
	Since the polynomial ring $\mathcal{K}[X]$ is a subring of a computable field $\mathcal{K}(X)$ (rational functions over $\mathcal{K}$),
	it holds for \aRA{$\mathcal{K}[X]$} as well.
	We therefore state the following theorem.
	\begin{theorem}\label{thmPolyRA}
		Let $\boldsymbol{Q}[X]=(\Q[X],+,\times)$.
		Functionality of \aRA{ $\boldsymbol{Q}[X]$} and equivalence of functional \aRA{ $\boldsymbol{Q}[X]$} are decidable.
	\end{theorem}
	
	The result of Theorem~\ref{thmPolyRA} can be extended to other algebras using simulations from algebra to algebra.
	Indeed, if $\aalg\simulatedBy\balg$, then any \aRA{\aalg} can be simulated by a \aRA{\balg},
	and problems of functionality and equivalence reduce from \aRA{\aalg} to \aRA{\balg}.
	
	\begin{corollary}\label{corRA-Equivalence}
		Let $\aalg$ be an algebra. If $\aalg\simulatedBy \boldsymbol{Q}[X]$, then
		functionality of \aRA{\aalg} and equivalence of functional \aRA{\aalg} are decidable.
	\end{corollary}
	\section{Unordered forests are simulated by polynomials}\label{secForests}
	In this section we will show that the unordered tree forests (and more generally -- the unordered forest algebra~\cite{BojanczykWalukiewicz08} that contains both forests and contexts with one hole)
	can be simulated in the sense of Definition~\ref{defSimulation} by polynomials with rational coefficients over a variable $x$ (noted $\Q[x]$) with the operations $+,\times$.
	This, combined with Corollary~\ref{corRA-Equivalence}, implies the decidability of functionality and equivalence for a class of \aRA{Forests}.
	We then prove that this class can express all MSO-transformations on unordered forests.
	
	An \emph{unordered tree} on a finite signature $\SigIn$ is an unranked tree (i.e.~every node can have arbitrarily many children), but the children of a node form an unordered multiset, rather than an ordered list.
	For example, the following figure displays two representations of the same unordered tree.
	An \emph{unordered forest} is a multiset of unordered trees.
	
	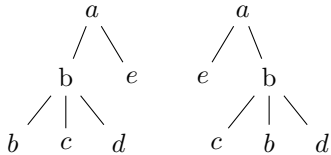
\begin{figure}[h]
		\begin{center}
			\begin{tikzpicture}[sibling distance=30, level distance=25]
			\node at (0,0) {$a$}
			child[sibling distance=20] {node {b}
				child {node {$b$}}
				child {node {$c$}}
				child {node {$d$}}
			}
			child {node {$e$}}
			;
			\node at (2,0) {$a$}
			child {node {$e$}}
			child[sibling distance=20] {node {b}
				child {node {$c$}}
				child {node {$b$}}
				child {node {$d$}}
			}
			;
			
			\end{tikzpicture}
		\end{center}
		\caption{Two representations of the same unordered tree}\label{figUnordered}
	\end{figure}
	
	Unordered forests can thus be defined as an algebra $\uf=(\setUF, +, \left\{\root_a\right\})$:
	\begin{enumerate}
		\item \setUF is the set of unordered forests, including $\emptyset$ -- the empty forest;
		\item the operations are:
		\begin{itemize}
			\item binary operation $+$ is the multiset addition,
			\item for each letter $a\in\SigIn$, unary operation $\root_a$: if $h=t_1+\dots+t_n$, then $\root_a(h)=a(\enum t1n)$.
		\end{itemize}
	\end{enumerate}
	
	In the rest of this paper, we will reason with a unary signature (and thus a unique \root operation). This is done without loss of generality, as unordered forests on a finite signature $\SigIn=\lbrace\enum a1n\rbrace$ can easily be encoded by forests on a unary signature. To express it as a polynomial simulation, we can say that $\alpha(\emptyset)=\emptyset$ , and that for all $1\leqslant i \leqslant n$
	$$\alpha(\root_{a_i}(h))=\root^i(\root(\emptyset)+\root(\alpha(h)))$$
	
	\subsection{Encoding forests into polynomials}
	
	This subsection's aim is to prove the following result:
	\begin{proposition}\label{propSimForests}
		$(\setUF,+,\root)$ is simulated by $(\Q[x],+,\times)$.
	\end{proposition}
	To this end we construct an injective homomorphism $\phi:\setUF\rightarrow\Rat[x]$.
	This $\phi$ associates injectively to each forest a rational polynomial $p$.
	It is important to check that two identical forests with different representations (as in Figure~\ref{figUnordered}) will not obtain different value by $\phi$.
	Furthermore, the operations $+, \root$ must be encoded as $\psi_+,\psi_r$, two polynomial functions in $(\Rat[x],+,\times)$,
	such that $\phi(h+h')=\psi_+(\phi(h),\phi(h'))$ and $\phi(\root(h))=\psi_r(\phi(h))$.
	
	Note that the term ``polynomial'' suffers here from semantic overload.
	We will take care to differentiate, on one hand, rational polynomials (i.e.~the elements of $\Rat[x]$, e.g. $2x-7$), denoted by variants on letters $p,q$,
	and on the other hand, polynomial functions on the algebra $(\Q[x],+,\times)$ (e.g. $\psi: (p,q)\mapsto q\times q + 2 p$), denoted by variants on the letter $\psi$.
	
	Since $+$ in \uf is both associative and commutative, we choose $\psi_+$ to be multiplication between rational polynomials: $\psi_+:(p,q)\mapsto p \times q$.
	This leaves \root to encode.
	To ensure that $\phi$ is injective, we would like to
	pick $\psi_r$ so that $\phi$ sends all $\root(h)$ to pairwise different irreducible polynomials.
	This is done by picking $\psi_r:p \mapsto 2+x\times p$ 
	and using the 
	\iflong
\else
	Eisenstein's
	\fi 
	criterion 
	\iflong
	\else
	with prime number 2: if a monic polynomial has all its nonleading coefficients divisible by 2, and the constant coefficient not divisible by 4, then this polynomial is irreducible over \Q.
	\fi
	\iflong
\else
	\fi
	From there we 
	define $\phi$ inductively: $\phi(\emptyset)=1$, $\phi(h+h')=\phi(h)\times\phi(h')$, and $\phi(\root(h))=2+x\times\phi(h)$.
	It is clear that $\phi$ respects the condition of polynomial simulation that any operation of \uf must be encoded as polynomial operation in $(\Rat[x],+,\times)$.
	\iflong
\else
	\fi
	
	%
	This leads directly to the proof of Proposition~\ref{propSimForests}: $\phi$ is a simulation from \uf to $\Q[x]$ as defined in Definition~\ref{defSimulation}.
	It is injective, the operation $+$ is encoded by the polynomial function $\psi_+:(p,q)\mapsto p \times q$,
	and the operation $\root$ is encoded  by the polynomial function $\psi_r:p \mapsto 2+x\times p$.
	\subsection{Extension to contexts}\label{subsecContext}
	The combination of Corollary~\ref{corRA-Equivalence} and Proposition~\ref{propSimForests} gives decidability results on the class of \aRA{\uf}.
	The transducers of this class read a ranked input, and manipulate registers with values in \setUF.
	As an example, an \aRA{\uf} can read a binary input, and output the unordered forests that it encodes in a ``First Child Next Sibling'' manner,
	that is to say the left child in the input corresponds to the child in the output, and the right child in the input corresponds to the brother in the output.
	Note that this is an adaptation of classical FCNS encoding of unranked \textbf{ordered} trees in binary trees, but where the order is forgotten.
	
	\begin{figure}[h]
		\begin{center}
			\begin{tikzpicture}[sibling distance=30, level distance=20]
			\node at (0,0) {$a$}
			child[sibling distance=30] {node {b}
				child[sibling distance=50] {node {$c$}
					child[sibling distance=30] {node {$\bot$}}
					child[sibling distance=30] {node {$d$}
						child[sibling distance=30] {node {$\bot$}}
						child[sibling distance=30] {node {$\bot$}}
					}
				}
				child[sibling distance=50] {node {$e$}
					child[sibling distance=30] {node {$\bot$}}
					child[sibling distance=30] {node {$f$}
						child[sibling distance=30] {node {$\bot$}}
						child[sibling distance=30] {node {$\bot$}}
					}
				}
			}
			child {node {$\bot$}}
			;
			\node at (6,0) {$a$}
			child {node {$b$}
				child {node {$c$}}
				child {node {$d$}}
			}
			child {node {$e$}}
			child {node {$f$}}
			;
			
			\draw [->,thick] (2,-1) -- (4,-1);
			
			\end{tikzpicture}
		\end{center}
		\caption{``FCNS'' decoding}\label{figFCNS}
	\end{figure}
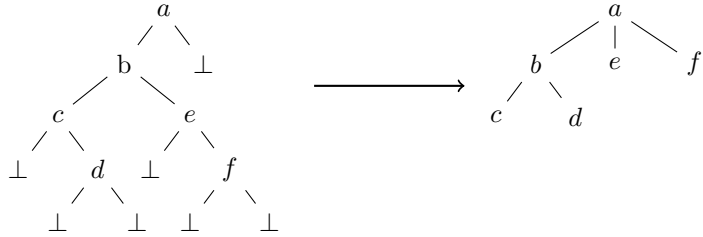
	This can be described by a one-state one-register \aRA{\uf} that uses rules of form
	$$(a,q,q)\rightarrow (q,(x,y)\mapsto \root_a(x)+y).$$
	
	However, \aRA{\uf} have their restriction: since $\root$ and $+$ are the only two operations allowed,
	registers can only store subtrees to be placed at the bottom of the output.
	This leaves the class without the ability to combine subtrees of its output as freely as the MSO logic does.
	As an example, it is impossible to create an \aRA{\uf} that,
	if given an input $f(u,v)$ where $u$ and $v$ are two unary subtrees,
	outputs the subtree $u$ above the subtree $v$ as shown in Figure~\ref{figConcat}.
	
	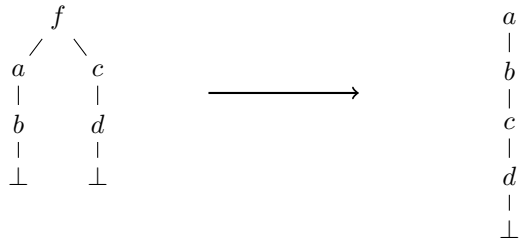
\begin{figure}[h]
		\begin{center}
			\begin{tikzpicture}[sibling distance=30, level distance=20]
			\node at (0,0) {$f$}
			child {node {$a$}
				child {node {$b$}
					child {node {$\bot$}}
				}
			}
			child {node {$c$}
				child {node {$d$}
					child {node {$\bot$}}
				}
			}
			;
			\node at (6,0) {$a$}
			child {node {$b$}
				child {node {$c$}
					child {node {$d$}
						child {node {$\bot$}}
					}
				}
			}
			;
			
			\draw [->,thick] (2,-1) -- (4,-1);
			
			\end{tikzpicture}
		\end{center}
		\caption{Subtree concatenation}\label{figConcat}
	\end{figure}
	
	To get a more general class of register automata, that can perform such superpositions,
	we need to allow registers to store contexts, rather than forests.
	While the use of the Hilbert Methods for algebras of general contexts remains a difficult and interesting open problem, we will show that forest contexts with at most one hole are simulated by polynomials of $\mathbb{Z}[x]$.%
	
	We use the unordered version of 2-sorted Forest Algebra~\cite{BojanczykWalukiewicz08}, consisting of unordered forests of trees and contexts with at most one hole. Since the previous subsection deals with an algebra of forests,
	to avoid confusion, we will call this the \emph{Unordered Context and Forest} algebra (noted \ucf).
	Using the definition of composition algebras, \ucf is a subset of $\uf[\circ]^{\mathsf{subs}}$,
	where we impose that the replacable variable $\circ$ occurs at most once.
	\begin{center}
		\newcommand{\treeA}[1]{
			node[Node] {}
			child {
				node[Node] {}
				child {
					#1
				}
			}
			child {
				node[Node] {}
				child {
					node[Node] {}
				}
				child {
					node[Node] {}
				}
			}
			child {
				node[Node] {}
			}
		}%
		\newcommand{\treeB}[1]{
			node[Node] (anchor) {}
			child {
				node[Node] {}
			}
			child {
				#1
			}
		}%
		$
		\left(
		\rule{0cm}{1.2cm}
		\ 
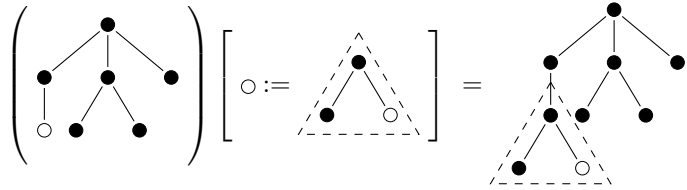
		\begin{tikzpicture}[baseline=-0.9cm,sibling distance=24, level distance=20]
		\draw \treeA{node[Hole] {}};
		\end{tikzpicture}
		\ 
		\right)
		\left[
		\ 
		\rule{0cm}{0.7cm}
		\tikz[baseline=-0.5ex]{ \node [Hole] {};} 
		:=
		\begin{tikzpicture}[baseline=-0.4cm,sibling distance=24, level distance=20]
		\draw \treeB{node[Hole] {}};
		\draw[dashed] ($(anchor.north) + (0,0.25)$) -- ++(0.8,-1.35) -- ++(-1.6, 0) -- cycle;
		\end{tikzpicture}
		\ 
		\right]
		\ 
		=
		\begin{tikzpicture}[baseline=-1.1cm,sibling distance=24, level distance=20]
		\draw \treeA{
			\treeB{node[Hole] {}}
		};
		\draw[dashed] ($(anchor.north) + (0,0.3)$) -- ++(0.8,-1.35) -- ++(-1.6, 0) -- cycle;
		\end{tikzpicture}
		$
	\end{center}
	On this algebra, we will show the following results:
	
	\begin{theorem}\label{thmUCFtoPol}
		$\ucf$ is simulated by $(\Q[x], +, \times)$.
	\end{theorem}
	
	\begin{corollary}\label{corUCFRA}
		Functionality of \aRA{\ucf} and equivalence of functional \aRA{\ucf} are decidable.
	\end{corollary}

	Lemma~\ref{lemSimSub} ensures that since $\uf\simulatedBy (\Rat[x],+,\times)$,
	then $\uf[\circ]^{\mathsf{subs}}\simulatedBy \Rat[x][y]^{\mathsf{subs}}$, i.e. $\Rat[x,y]$ where only $y$ can be substituted.
	\ucf is the restriction of $\uf[\circ]^{\mathsf{subs}}$ to its elements with at most one occurrence of $\circ$.
	This forms a 2-sorted algebra. We consider its natural match in $\Rat[x][y]^{\mathsf{subs}}$:
	Let $\aalg$ be the 2-sorted algebra:
	\begin{itemize}
		\item The universe is $\aalgset:=\{p(x)+yq(x): p,q\in\Q[x]\}\subseteq \Q[x,y]$.
		\item The types are $\aalgset_0:= \Q[x], \aalgset_1:=\aalgset\setminus \aalgset_0$.
		\item The operations are:
		\begin{itemize} 
			\item multiplication, defined only on pairs of types: $(0,0), (0,1), (1,0)$,
			\item $(-)[y:=(-)].$
		\end{itemize}
	\end{itemize}
	
	%
	\begin{lemma}\label{lemUCFtoTypePol}
		$\ucf$ is simulated by $\aalg$.
	\end{lemma}
	\begin{proof}
		We call $\alpha:\uf[\circ]^{\mathsf{subs}}\rightarrow\Rat[x][y]^{\mathsf{subs}}$
		the homomorphism obtained by extending last subsection's $\phi$ with mapping the substitution variable $\circ$ to $y$. 
		We restrict $\alpha$ to terms
		with at most one occurrence of $\circ$.
		The image of $h\in\ucf$ will then be a term of $\Rat[x][y]^{\mathsf{subs}}$ with at most one occurrence of $y$.
		If $\circ$ never appears in $h$, then $y$ never appears in $\alpha(h)$, thus $\alpha(h)\in \aalg_0$.
		If $\circ$ appears once in $h$, then $y$ appears once in $\alpha(h)$, thus $\alpha(h)\in \aalg_1$.
		%
	\end{proof}
	
	We now prove that \aalg is simulated by $\Q[x]$ \emph{without} substitution.
	
	\begin{lemma}\label{lemTypePoltoPol}
		$\aalg$ is simulated by $(\Q[x], +, \times)$.
	\end{lemma} 
	\begin{proof}
		We will use encoding of $\aalg$ in $\Q[x]\times \Q[x]$ given by 
		\iflong
\else
		$p(x)+yq(x)\mapsto (p, q).$
		\fi
		Provided this, we encode operations $+,\times, (-)[y:=(-)]$ in a straightforward manner.
		\iflong
\else
		For example, for the composition operation in \aalg, we see that
		$\comp{p(x)+yq(x)}{y}{p'(x)+yq'(x)}$ is equal to $p(x)+p'(x)q(x)+yq'(x)q(x)$.
		Hence, in pairs of \Q[x], $\comp-y-$ is encoded by $\psi_{\comp-y-}:(p,q,p',q')\mapsto(p+p'q,q'q).$
		\fi
	\end{proof}
	
	Since $\simulatedBy$ is a transitive relation, Lemma~\ref{lemTypePoltoPol}
	and Lemma~\ref{lemUCFtoTypePol} give Theorem~\ref{thmUCFtoPol}.
	Once Theorem~\ref{thmUCFtoPol} is proven, Corollary~\ref{corRA-Equivalence} gives Corollary~\ref{corUCFRA}.

	
	Note that this proof extends to contexts with a bounded number of holes.
	We can add $N$ substitution variables $\enum {\circ}1N$ to $\uf$.
	Lemma~\ref{lemSimSub} gives a homomorphism $\alpha_N$ that ensures
	$\uf[\enum{\circ}1N]^{\mathsf{subs}}\simulatedBy \Q[x][\enum y1N]^{\mathsf{subs}}$.
	One could then define contexts with at most $M$ occurrences 
	of
	variables \ucfm.
	In a manner similar to Lemma~\ref{lemUCFtoTypePol}, we can find a finitely-sorted algebra that contains $\alpha_N(\ucfm)$,
	i.e. an algebra of 
	all polynomials of $\Q[x][\enum y1N]$ with a degree $\leqslant M$ regarding the variables $\enum y1N$.
	Then, in a manner similar to Lemma~\ref{lemTypePoltoPol}, we can show that finite degree composition can be encoded in $(\Q[x],+,\times)$.
	%
	%
	\begin{corollary}\label{corUCFMRA}
		$\ucfm$ is simulated by $(\Q[x], +, \times)$.
		Functionality of \aRA{\ucfm} and equivalence of functional \aRA{\ucfm} are decidable.
	\end{corollary}
	
	
	

	\iflong
	
	\else
	\subsection{Encompassing of MSO}
	
	Corollary~\ref{corUCFRA} gives decidability results on the class of \aRA{\ucf}.
	We motivated this class as a relevant extension of \aRA{\uf} by exhibiting a transformation (see Figure~\ref{figConcat})
	that required contexts to be expressed.
	However, this class is not immediately relevant in its properties or expressiveness.
	In this section, we prove that \aRA{\ucf} can express strictly more than all MSO-definable transformations on unordered trees.
	Note that \aRA{\ucf} define functions from \emph{binary ordered} trees to \setUCF, not from \setUF to \setUF.
	We say that an \aRA{\ucf} expresses a function $f:\setUF\rightarrow\setUF$ if
	for a binary tree $t$ that is the ``FCNS'' encoding of a forest $h$, its image for the tree $t$ is $f(h)$.
	
	We briefly present a definition of MSO formulae and transformations.
	More complete definitions exist elsewhere in the literature (e.g. \cite{Courcelle94}).
	
	The syntax of monadic second order logic (MSO) is:
	$$\phi:=\phi\land\phi \;|\; \lnot\phi \;|\; \exists x \phi \;|\; \exists X \phi \;|\; x\in X$$
	where lower cases $x$ are node variables, and upper cases $X$ are set variables.
	This syntax is enriched by different relations to describe the structure of the objects we consider:
	\begin{itemize}
		\item For binary trees (BT), we add two relations $\child_L(x,y)$ and $\child_R(x,y)$ that express that $y$ is the left child (resp.~right child) of $x$.
		\item For unranked ordered forests (OF), we add $\fc(x,y)$, that expresses that $y$ is the first child of $x$, and $\ns(x,y)$ that express that $y$ is the brother directly to the right of $x$.
		\item For unranked unordered forests (UF), we only add the relation $\child(x,y)$, that expresses that $y$ is a child of $x$. The relation ``Sibling'' would only be syntactic sugar.
	\end{itemize}
	
	An \emph{MSO-definable transformation} with $n$ copies is a transformation that for each input node $x$, makes $n$ output nodes \enum x1n.
	The presence or absence of an edge in the output are dictated by formulae defining the transformation.
	A MSO-definable transformation is characterized by its formulae $\varphi_{\mathcal{R},i,j}$ for each $1\leqslant i,j \leqslant n$,
	and each structure relation $\mathcal R$ (e.g. $\fc$ and $\ns$ if the output is ordered forests).
	
	For example, if one wanted to reverse left and right children in binary trees, this would be a transformation definable in \msotype {BT}{BT} with one copy,
	where $\varphi_{\child_L,1,1}(x,y)=\child_R(x,y)$, i.e.~$y_1$ is $x_1$'s left child in the output iff $y$ was $x$'s right child in the input,
	and conversely $\varphi_{\child_R,1,1}(x,y)=\child_L(x,y)$.
	
	We note that this definition can express transformations between any two tree algebras.
	For example, the ``FCNS'' decoding of Figure~\ref{figFCNS} can be encoded in MSO from binary trees to \setUF.
	Since we will use different combinations of input-output in this part, we introduce the notation
	\msotype {\bullet}{\bullet} to denote MSO from one type of trees to the other.
	\msotype {BT}{OF} designs MSO-definable functions from binary trees to ordered forests,
	and \msotype {UF}{UF} designs MSO-definable functions from unordered forests to unordered forests.
	
	\begin{proposition}\label{propMSOtoaRA}
		Every function of $\msotype {UF}{UF}$ can be described by an \aRA{\ucf}.
	\end{proposition}
	
	The proof we provide to show this Proposition has three arguments:
	\begin{enumerate}
		\item \msotype {UF}{UF} can be represented by functions of \msotype {BT}{OF}.
		\item Bottom-UP Streaming Tree Transducers (STT)~\cite{AlurDAntoni17} describe all functions of \msotype {BT}{OF}.
		\item Bottom-Up STT can be expressed as register automata.
	\end{enumerate}
	\partitle{From \msotype {UF}{UF} to \msotype {BT}{OF}}
	We say that a binary tree $t$ \emph{represents} an unordered forest $h$
	if the ``FCNS'' decoding of $t$ as represented in Figure~\ref{figFCNS} is $h$. Note that $t$ is not unique for $h$,
	but every $t$ represents a unique $h$.
	Similarily, we can say that an unranked ordered forest $h$ \emph{represents} an unordered forest $h'$
	if by forgetting the siblings' order in $h$, we get $h'$. Once again such an $h$ is not unique for $h'$,
	but every $h$ represents a unique $h'$.
	We can extend this notion to MSO transformation.
	
	\begin{definition}
		A function $f\in \msotype {BT}{OF}$
		\emph{represents} a function $f'\in \msotype {UF}{UF}$ if:
		\begin{itemize}
			\item For every unordered forest $h$ such that $f'$ is defined over $h$,
			then there exists at least one binary tree $t$ such that
			$t$ represents $h$, and $f$ is defined over $t$
			\item For every binary tree $t$ such that $f$ is defined over $t$,
			$f'$ is defined over the unordered forests $h$ such that $t$ represents $h$,
			and $f(t)$ represents $f'(h)$.
		\end{itemize}
	\end{definition}
	
	Once again such an $f$ is not unique for $f'$,
	but every $f$ represents a unique $f'$.
	Furthermore, it is always possible to find a representant $f$ for an MSO-definable function $f'$.
	
	\begin{lemma}\label{lemMSOtrad}
		If $f'\in \msotype {UF}{UF}$, then there exists $f\in \msotype {BT}{OF}$ that represents $f'$.
	\end{lemma}
	
	\begin{proof}
		We start by encoding the input, transforming $f'$ into a function of \msotype {BT}{UF}.
		To modify $f'$ so that it transforms trees that represent $h$ into $f'(h)$, one has to replace every occurence
		of $\child(x,y)$ into $\varphi_{\child,i,j}$ by its ``FCNS'' encoding, i.e.~$\exists z \mid \child_L(x,z)\land\child^*_L(z,y)$.
		
		Encoding the output requires to change $\varphi_{\child,i,j}$ into two relations
		$\varphi_{\fc,i,j}$ and $\varphi_{\ns,i,j}$, i.e.~to artificially order the siblings of the output forest.
		To that effect, we note that from the BT-formula $\varphi_{\child,i,j}$, one can describe in BT-MSO
		a set $S_{x,i,j}=\lbrace y \mid \varphi_{\child,i,j}(x,y)\rbrace$.
		Since $S_{x,i,j}$ is a set of input nodes of a binary tree, it is totally ordered by their occurrence in the infix run.
		This order can be expressed as a BT-MSO relation.
		We can then decide to order all the children $y_j$ of the output node $x_i$.
		
		To find the first child of a node in the output, we say that
		$\varphi_{\fc,i,j}(x,y)$ if $j$ is the first index where $S_{x,i,j}\neq\emptyset$
		and $y$ is its first element.
		Similarily, to find the next sibling of a node in the output, we say that 
		$\varphi_{\ns,i,j}(y,z)$ if $\exists x \mid \varphi_{\child,k,i}(x,y)\land\varphi_{\child,k,j}(x,z)$,
		and either $i=j$ and $y,z$ are consecutive elements of $S_{x,k,i}$,
		or $y$ is the last element of $S_{x,k,i}$, $z$ is the first element of $S_{x,k,j}$,
		and $j$ is the first index bigger than $i$ such that $S_{x,k,j}\neq\emptyset$.
	\end{proof}
	
	\partitle{From \msotype {BT}{OF} to STT to RA}
	The next step is to use an existing result from the literature~\cite{AlurDAntoni17}
	that describes a model of transducers that describes all $\msotype {BT}{OF}$.
	The formalism in question are \emph{Streaming Tree Transducers} (STT).
	A STT is an automaton on nested words (words representing trees) that maintains a stack of register configurations.
	The nesting of the words dictates how this stack behaves:
	each opening letter ${<}a$ stores the current variable values in the stack to start with fresh ones,
	then each closing letter $a{>}$ uses the current variable values and the top of the stack to generate new values for the registers.
	
	In~\cite{AlurDAntoni17}, STT are limited to linear functions for the update of their value.
	Furthermore, the paper proves that without loss of expression, one can consider
	\emph{Bottom-Up STT}, where reading an opening symbol ${<}a$ resets the state as well as the registers.
	On such STT, the behavior of a STT reading the nested word of a subtree does not depend on what occurs before or after,
	and its computation behaves like a register automaton reading a tree in a bottom-up manner.
	We will consider the class of Bottom-Up STT that read nested representations of binary trees (Bottom-Up BT STT).
	
	\begin{proposition}\label{propMSOtoSTT}
		Every function of \msotype {BT}{UF} is described by a Bottom-Up BT STT.
	\end{proposition}
	
	\begin{proposition}\label{propSTTtoaRA}
		Every function of a Bottom-Up BT STT is described by an \aRA{OCF}.
	\end{proposition}
	
	Proposition~\ref{propMSOtoSTT} comes directly from Theorems 3.7 and 4.6 of~\cite{AlurDAntoni17}: 3.7 explains Bottom-Up BT-STT are as powerful as general BT-STT, and 4.6 states that STT can describe any function of \msotype {BT}{UF}.
	Proposition~\ref{propSTTtoaRA} is not directly proven in~\cite{AlurDAntoni17}
	but their definition of Bottom-Up STT is made specifically to that end. We provide more details in the appendix.
	
	\partitle{End of Proof}
	To turn that \aRA{OCF} into an \aRA{\ucf}, we just have to change the ordered concatenation of OCF
	to the unordered concatenation of \ucf. By combining Lemma~\ref{lemMSOtrad},
	Proposition~\ref{propMSOtoSTT} and Proposition~\ref{propSTTtoaRA},
	we conclude our proof of Proposition~\ref{propMSOtoaRA}.
	
	We note that every MSO-definable function can be described by a \aRA{\ucf}, however the converse is not true;
	consider a function that creates output of exponential size (whereas MSO can only describe functions of linear size increase).
	Consider unary input trees of form $\root_a^n(\bot)$, and a 1-counter \aRA{\ucf} with rules
	$\bot\rightarrow q(\root())$, and $a(q(h))\rightarrow h+h$.
	The image doubles in size each time a symbol is read.
	Unsurprisingly, this counterexample uses the copyful nature of \aRA{\ucf},
	as copyless restrictions tend to limit the expressivity power of register automata to MSO classes~\cite{AlurCerny10,AlurDAntoni17}.
	
	\fi
	\section {On decidability of MTT equivalence. Equivalence of \aRA{polynomials} with composition is undecidable}\label{secPolySub}
	In this section, we tru to use the ``Hilbert Method'' to study the equivalence problem on Macro Tree Transducers (MTT)~\cite{EngelfrietVogler85}.
	MTT have numerous definitions. For this paper, we will consider them to be register automata on an algebra of ranked trees
	with an operation of substitution on the leaves; observe this is exactly \aRA{$\orderedTreesAlg[X]^{\mathsf{subs}}$}.
	The algebra $\orderedTreesAlg$ (ranked trees without substitution on the leaves) can be simulated by words with concatenation (via nested word encoding).
	Words with concatenation can be encoded by $\Q$ (see, for example, the proof of Corollary~10.11~\cite{toolbox}).
	Thus, $\orderedTreesAlg\simulatedBy\Q$.
	Finally, by Lemma~\ref{lemSimSub}, we have that $\orderedTreesAlg[X]^{\mathsf{subs}}\simulatedBy\Q[X]^{\mathsf{subs}}$.
	This means that if equivalence is decidable for \aRA{$\Q[X]^{\mathsf{subs}}$}, then MTT equivalence is decidable.
	Unfortunately,  we will show that even with one variable $x$, the register automata of \aRA{\compalg} have undecidable functionality and equivalence:
	\begin{theorem}\label{thmUndecidability}
		The functionality problem for \aRA\compalg and equivalence problem for functional \aRA\compalg are undecidable.
	\end{theorem}
	
	%
	
	%
	\iflong
\else
	We prove this undecidability result
	by reducing the reachability problem for 2-counter machines
	to the equivalence problem on deterministic \aRA{\compalg}
	with a monadic input (i.e. that reads words rather than trees).
	This means that the actual theorem we prove is slightly more powerful
	than Theorem~\ref{thmUndecidability}. Its full extent is described in
	Theorem~\ref{thmFullUndecidability}.
	We recall the definition of a 2-counter machine.
	\fi
	\begin{definition}
		A 2-counter machine (2CM) is a pair $\M=(Q,\rul)$, where:
		\begin{itemize}
			\item $Q$ is a finite set of states,
			\item $\rul:Q\times\lbrace 0,1 \rbrace\times\lbrace 0,1 \rbrace\rightarrow Q\times\lbrace -1,0,1 \rbrace\times\lbrace -1,0,1 \rbrace$ is a \emph{total} transition function.
		\end{itemize}
		A configuration of \M is a triplet of one state and two nonnegative integer values (or counters) $(q,c_1,c_2)\in \Q\times\Nat\times\Nat$.
		We describe how to use transitions between configurations: $(q,c_1,c_2)\rightarrow(q',c'_1,c'_2)$ if there exists $(q,b_1,b_2)\rightarrow(q',d_1,d_2)$ in $\rul$ such that for $i\in\lbrace 1,2\rbrace$:
		$c_i=0\iff b_i=0$ and $c'_i=c_i+d_i$.
		Note that to ensure that no register wrongfully goes into the negative, we assume wlog that if there exists $(q,b_1,b_2)\rightarrow(q',d_1,d_2)$ in $\rul$,
		then $d_i=-1\implies b_i=1$ (i.e. we can only decrease a non-zero counter).
		
	\end{definition}
	
	The \emph{2CM reachability problem} can be expressed as such: starting from an initial configuration $(q_0,0,0)$, can we access the state $q_f\in Q$,
	i.e. is there a configuration $(q_f,c_1,c_2)$ such that $(q_0,0,0)\rightarrow^*(q_f,c_1,c_2)$.
	It is well known that 2CM reachability is undecidable.
	
	\partitle{Reduction of 2CM Reachability to \aRA{\compalg} Equivalence}
	Let $\M=(Q,\rul)$ be a $n$-states 2CM. We rename its states $Q=\lbrace 1,\dots,n\rbrace$.
	We consider the 2CM reachability problem in \M from state $1$ to state $n$.
	
	We simulate this machine with a \aRA{\compalg} $\M'$.
	It will have only one state $q_{\M'}$: the configurations $(q,c_1,c_2)$ of $M$ will be encoded in 3 registers of $\M'$.
	It will work on a signature $\bot\cup\delta$, where $\bot$ is of rank 0 and every transition $(q,b_1,b_2)\rightarrow(q',d_1,d_2)$ of $\rul$ is a symbol of rank $1$.
	Intuitively, reading a symbol $(q,b_1,b_2)\rightarrow(q',d_1,d_2)$ in $\M'$ models executing this transition in $\M$.
	The automaton will have 6 registers: 3 to encode the configurations of $\M'$ and 3 containing auxiliary polynomials useful to test if the input sequence of transitions describes a valid computation in $M$.
	
	\noindent We encode the configurations $(q,n_1,n_2)$ in 3 registers as follows:
	\begin{itemize}
		\item register $r_{q}$ holds the (number of the) current state.
		\item registers $r_{c_1}$, $r_{c_2}$ hold $n_1, n_2$ -- the current values of the counters.
	\end{itemize}
	
	We now encode transitions in \M as register operations in \M'.
	When reading a transition $(i,b_1,b_2)\rightarrow(j,d_1,d_2)$, the update of the configuration is natural.
	However, we must ensure that we are allowed to use this transition in the current configuration.
	
	To this end, we keep in \M' a witness register. Its value will be 0 if and only if the sequence of transitions read as an input does not constitute a valid path in \M.
	To update such a register, when a transition $(i,b_1,b_2)\rightarrow(j,d_1,d_2)$ is read, we need to check that $r_{q}=i$ and that $r_{c_i}=0\iff b_i=0$.
	
	For the state $q\in\lbrace1,\dots,n\rbrace$, we design \ptestq i, a polynomial such that $\ptestq i(i)\neq0$ and for every other value $1\leqslant j \leqslant n$, $\ptestq i(j)=0$:
	$\ptestq i(x):=\prod_{1\leqslant j \leqslant n}^{i\neq j}{(x-j)}$.
	This approach cannot work for counters, as there is no absolute bound to their value.
	To remedy that problem, we will design \emph{for each $m$} a polynomial \ptestc m such that $\ptestc m(0)\neq 0$ and for every other value, $1\leqslant j \leqslant m$, $\ptestc m(j)=0$.
	$\ptestc m(x):=\prod_{1\leqslant j \leqslant m}{(x-j)}$.
	Intuitively, \ptestc m works as a test for counters in the $m$-th step of \M, since counters $c_1,c_2$ cannot exceed the value $m$ at that point.
	This means that \ptestc m will have to be stored and updated in a register of its own.
	To this end, we introduce the last three registers of $\M'$:
	\begin{itemize}
		\item the register $r_+$. After $m$ steps, $r_+=m$.
		\item the register $r_{zt}$. After $m$ steps, $r_{zt}=\ptestc m$.
		\item the witness register $r_w$. After $m$ steps, $r_w\neq 0\iff$ we read a valid path in $\M$.
	\end{itemize}
	
	%
	We describe how to update the registers of $\M'$ when reading an input symbol $(i,b_1,b_2)\rightarrow(j,d_1,d_2)$.
	Note that according to our definition of \aRA{\compalg}, the new values $\ol{r}'$ are computed as a function of the old value of $\ol r$.
	This means that any value on the right of the assignation symbol $\leftarrow$ is the value before reading $(i,b_1,b_2)\rightarrow(j,d_1,d_2)$.
	\begin{itemize}
		\item $r_q\leftarrow j$, $r_{c_1}\leftarrow r_{c_1}+d_1$, $r_{c_2}\leftarrow r_{c_2}+d_2$,
		\item $r_+\leftarrow r_++1$, $r_{zt}\leftarrow r_{zt}\times (x-r_+-1)$,
		\item $r_w\leftarrow r_w \times T_q \times T_{1} \times T_{2}$, where:
		\begin{itemize}
			\item $T_q=\comp {\ptestq i}x{r_q}$,
			\item for $i\in\lbrace 1, 2 \rbrace$, $T_i=\begin{cases}
			r_{c_i} \text{ if }b_i=1, \\
			\comp{r_{zt}}x{r_{c_i}} \text{ if }b_i=0.
			\end{cases}$
		\end{itemize}
	\end{itemize}
	
	%
	This update strategy ensures that each counter does what we established its role to be.
	The only register for which this is not trivial is $r_w$. We show that
	$r_w=0$ if and only if we failed to read a proper path in $\M$.
	
	We proceed by induction on the number of steps.
	The induction hypothesis is that a mistake happened before the $m$-th step if and only if $r_w=0$ before reading the $m$-th symbol.
	If such is the case, $r_w$ will stay at zero for every subsequent step, as the new value of $r_w$ is always a multiple of the previous ones.
	If the error occurs exactly at the $m$-th step, it means that the $m$-th letter of the input
	was a transition $(i,b_1,b_2)\rightarrow(j,d_1,d_2)$,
	but $r_q$ was not $i$ (and hence $T_q=\comp {\ptestq i}x{r_q}=0$),
	or that for this transition to apply we need the counter $c_i$ to be $0$ when it was not (or conversely assumed it $>0$ when it was $0$).
	This last case is caught by $T_i$. 
	By using $T_i=r_{c_i}$, we have $T_i=0$ exactly when we were wrong.
	If $b_i=0$ then we assume $c_i=0$. We know that $c_i\leq m$, where $m$ is the number of step taken. 
	By using $T_i=\comp{r_{zt}}x{r_{c_i}}=\ptestc m(c_i)$, we have that $T_i=0$ exactly when $0<c_i\leq m$.
	
	%
	The final step of the reduction comes by picking the output function for the only state of $\M'$.
	We pick $f(\ol r):=\comp {\ptestq i}x{r_q} \times r_w$.
	The only way for the output to not be 0 is if $r_q$ ends in $n$ (i.e. we reached state $n$)
	and if $r_w\neq 0$ (i.e. we used a valid path). In other words, the following Lemma holds.
	
	\begin{lemma}\label{lemMisMprime}
		$\sem{\M'}$ is the constant 0 function if and only if n is not reachable from 1 in $\M$
	\end{lemma}
	
	By comparing $\M'$ to a \aRA\compalg $M_0$ performing the constant 0 function,
	we get that deciding equivalence on functional \aRA\compalg would allow to decide 2CM Reachability.
	Similarly, running nondeterministically $\M'$ and $M_0$,
	we get that deciding functionality on \aRA\compalg would allow to decide 2CM Reachability.
	This leads to the proof of Theorem~\ref{thmUndecidability}.
	More than that, we show a more thorough result:
	
	\begin{theorem}\label{thmFullUndecidability}
		The equivalence problem for deterministic \aRA\compalg is undecidable, even with a monadic input alphabet.
		The functionality problem for \aRA\compalg is undecidable, even with a one-letter monadic alphabet.
	\end{theorem}
	
	The first point is given directly by the point above: $\M'$ is a deterministic \aRA\compalg on a monadic alphabet $\bot\cup\delta$, that is to say, the input of $M'$ is a word where each letter is an element of $\delta$.
	
	For the second point, we imagine a slight alteration of this proof where the input alphabet is $\lbrace a,\bot\rbrace$ where $\bot$ is of rank 0 and $a$ of rank 1, that is to say, the input of the \aRA\compalg would be a word of form $a^k$.
	In this new version, $M'$ is no longer deterministic, but guesses each time what transition of $M$ to emulate.
	When $M'$ reads $a^k$, it either guesses a correct path of length $k$, or makes a mistake and returns 0.
	$M'$ is functional iff it cannot guess a run that produces something else than 0, i.e. iff n is not reachable from 1 in $\M$.

	\section{Conclusion}
	
	We use ``Hilbert Methods'' to study equivalence problems on register automata.
	To apply these methods to register automata on contexts, we consider algebras with a substitution operation.
	To show the decidability of equivalence on \aRA{\ucf}, a class that subsumes MSO-definable transformations in unordered forests,
	we use the fact that bounded degree substitution can be encoded into $+,\times$ in $\Q[X]$.
	However, when applying the same method to Macro Tree Transducers,
	we are led to consider register automata on $\Q[X]^{\mathsf{subs}}$,
	whose equivalence we prove to be undecidable.
	In essence, for the ``Hilbert Methods'' we consider to provide positive results,
	it seems necessary to limit the use of composition.
	
	Future developments of this work could then consist of finding other acceptable restrictions
	on the use of composition in $\Q[X]$ that still allows for decidability results in register automata.
	Another possible avenue is to use the properties of $\simulatedBy$ to prove negative results:
	if $\aalg\simulatedBy\balg$, and register automata have undecidable problems in \aalg,
	then this negative results propagates to \balg.
	Finally, ``Hilbert Methods'' can apply to a huge variety of algebras (e.g. \ucf in this paper or $\Q^n$ in~\cite{BenediktDSW17}).
	They provide decidability results on register automata on algebras with nontrivial structure properties like commutativity of operations (e.g. children in \ucf) that make the usual methods to decide equivalence difficult to apply. 
	
	\textbf{Acknowledgments} {We would like to thank Miko{\l}aj Boja\'nczyk, from the University of Warsaw, for introducing us to the topic of using the Hilbert Method to decide equivalence of register automata, and for his active participation in the finding of this result and the writing of this paper.\\
		This work was supported by the European Research Council (ERC) under the European Union’s Horizon 2020 research and innovation programme (ERC consolidator grant LIPA, agreement no. 683080) and partially supported by the NCN grant 2016/21/B/ST6/01505.}
	\bibliographystyle{plainurl}
	\bibliography{bib}

\begin{thebibliography}{10}

\bibitem{AlbertLawrence85}
Michael~H. Albert and J.~Lawrence.
\newblock A proof of ehrenfeucht's conjecture.
\newblock {\em Theor. Comput. Sci.}, 41:121--123, 1985.

\bibitem{AlurCerny10}
Rajeev Alur and Pavol Cern{\'{y}}.
\newblock Expressiveness of streaming string transducers.
\newblock In {\em {FSTTCS}}, volume~8 of {\em LIPIcs}, pages 1--12. Schloss
  Dagstuhl - Leibniz-Zentrum fuer Informatik, 2010.

\bibitem{AlurDAntoni17}
Rajeev Alur and Loris D'Antoni.
\newblock Streaming tree transducers.
\newblock {\em J. {ACM}}, 64(5):31:1--31:55, 2017.
\newblock URL: \url{http://doi.acm.org/10.1145/3092842}, \href
  {http://dx.doi.org/10.1145/3092842} {\path{doi:10.1145/3092842}}.

\bibitem{AlurDDRY13}
Rajeev Alur, Loris D'Antoni, Jyotirmoy~V. Deshmukh, Mukund Raghothaman, and
  Yifei Yuan.
\newblock Regular functions and cost register automata.
\newblock In {\em {LICS}}, pages 13--22. {IEEE} Computer Society, 2013.

\bibitem{BenediktDSW17}
Michael Benedikt, Timothy Duff, Aditya Sharad, and James Worrell.
\newblock Polynomial automata: Zeroness and applications.
\newblock In {\em {LICS}}, pages 1--12. {IEEE} Computer Society, 2017.

\bibitem{toolbox}
Miko{\l}aj Boja{\'{n}}czyk.
\newblock {\em {Automata toolbox}}.
\newblock May 2018.
\newblock URL:
  \url{https://www.mimuw.edu.pl/~bojan/paper/automata-toolbox-book}.

\bibitem{BojanczykWalukiewicz08}
Miko{\l}aj Boja{\'{n}}czyk and Igor Walukiewicz.
\newblock Forest algebras.
\newblock In {\em Logic and Automata}, volume~2 of {\em Texts in Logic and
  Games}, pages 107--132. Amsterdam University Press, 2008.

\bibitem{Choffrut03}
Christian Choffrut.
\newblock Minimizing subsequential transducers: a survey.
\newblock {\em Theor. Comput. Sci.}, 292(1):131--143, 2003.

\bibitem{Courcelle94}
Bruno Courcelle.
\newblock Monadic second-order definable graph transductions: {A} survey.
\newblock {\em Theor. Comput. Sci.}, 126(1):53--75, 1994.

\bibitem{EngelfrietManeth03}
Joost Engelfriet and Sebastian Maneth.
\newblock Macro tree translations of linear size increase are {MSO} definable.
\newblock {\em {SIAM} J. Comput.}, 32(4):950--1006, 2003.

\bibitem{EngelfrietVogler85}
Joost Engelfriet and Heiko Vogler.
\newblock Macro tree transducers.
\newblock {\em J. Comput. Syst. Sci.}, 31(1):71--146, 1985.

\bibitem{Honkala00}
Juha Honkala.
\newblock A short solution for the {HDT0L} sequence equivalence problem.
\newblock {\em Theor. Comput. Sci.}, 244(1-2):267--270, 2000.

\bibitem{CulikKarhumaki86}
Karel~Culik II and Juhani Karhum{\"{a}}ki.
\newblock The equivalence of finite valued transducers (on {HDT0L} languages)
  is decidable.
\newblock {\em Theor. Comput. Sci.}, 47(3):71--84, 1986.

\bibitem{Maneth15}
Sebastian Maneth.
\newblock A survey on decidable equivalence problems for tree transducers.
\newblock {\em Int. J. Found. Comput. Sci.}, 26(8):1069--1100, 2015.
\newblock URL: \url{https://doi.org/10.1142/S0129054115400134}, \href
  {http://dx.doi.org/10.1142/S0129054115400134}
  {\path{doi:10.1142/S0129054115400134}}.

\bibitem{Rounds70}
William~C. Rounds.
\newblock Mappings and grammars on trees.
\newblock {\em Mathematical Systems Theory}, 4(3):257--287, 1970.

\bibitem{SeidlMK15}
Helmut Seidl, Sebastian Maneth, and Gregor Kemper.
\newblock Equivalence of deterministic top-down tree-to-string transducers is
  decidable.
\newblock In {\em {FOCS}}, pages 943--962. {IEEE} Computer Society, 2015.

\end{thebibliography}
	\iflong
	\else
	\newpage
	\section*{Appendix}
	
	\partitle{Bottom-Up Streaming Tree Transducers}
	We go more into detail into \emph{Streaming Tree Transducers} (STT), that read and output nested words.
	This formalism is central to the proof of Proposition~\ref{propMSOtoSTT}, and of Proposition~\ref{propSTTtoaRA}.
	
	In intuition, an STT works with a configuration composed of a state, a finite number of typed variables (or registers)
	that contains nested words with at most one occurrence of a context symbol (this corresponds to the Ordered Forest Algebra (OCF) in the sense of~\cite{BojanczykWalukiewicz08}),
	and a stack containing pairs of stack symbols and variable valuations.
	The nesting of the words dictates how this stack behaves:
	each opening letter ${<}a$ stores the current variable values in the stack to start with fresh ones,
	then each closing letter $a{>}$ uses the current variable values and the top of the stack to generate new values for the registers.
	The operations on nested words that can be performed in such cases correspond to polynomial operations on OCF:
	one can use concatenation, context application (which translates directly into OCF),
	or use a constant nested word, that can be simulated by the roots and concatenation:
	$({<}a\; \circ\; a{>})[\circ:=r\cdot r']{<}b \; b{>}$ can be seen in forests as $\root_a(r+r')+\root_b()$. 
	
	The general definitions are available on~\cite{AlurDAntoni17}.
	We use specifically \emph{bottom-up} STT, where reading an opening symbol ${<}a$ resets the state as well as the registers.
	On such STT, the behavior of a STT reading the nested word of a subtree does not depend on what occurs before or after.
	The original paper also imposes a single-use restriction, to ensure each operation can use each register only once. We can keep this restriction, but will not need it.
	We add a few restrictions to this model:
	\begin{itemize}
		\item We do not allow letters beyond nesting letters. In the language of~\cite{AlurDAntoni17} this means we ignore internal transitions.
		\item The input domain is a language of nested words of binary trees.
	\end{itemize}
	
	We will call this subclass Bottom-Up BT STT.
	The first result we need comes directly from the results of~\cite{AlurDAntoni17}
	that states that single-use STT (even limited to bottom-up) describe exactly MSO functions on nested words:
	
	\begin{proposition}\label{propMSOtoSTTannex}
		Every function of \msotype {BT}{UF} is described by a Bottom-Up BT STT.
	\end{proposition}
	
	\partitle{From STT to \aRA{\ucf}}
	To complete the proof, we show that if a Bottom-Up BT STT describes $f$,
	then we can find a \aRA{\ucf} that describes the function $f'$ that $f$ represents, by forgetting the order in the output.
	\begin{proposition}\label{propSTTtoaRAannex}
		Every function of a Bottom-Up BT STT is described by an \aRA{OCF}.
	\end{proposition}
	\begin{proof}
		We propose in a figure below the run of a bottom-up STT in a tree $c(t,t')$.
		The subtrees $t$ and $t'$ are of root $a$ and $b$.
		The second line corresponds to its configuration (state $q$, register values $\ol r$),
		while the third line keeps track of the top symbol of the STT's stack.
		The state $q_0$ and register valuation $\ol{r_0}$ are respectively the initial state and register values.
		The symbol that was at the top of the stack when reaching ${<}c$ is denoted as $\lambda$.
		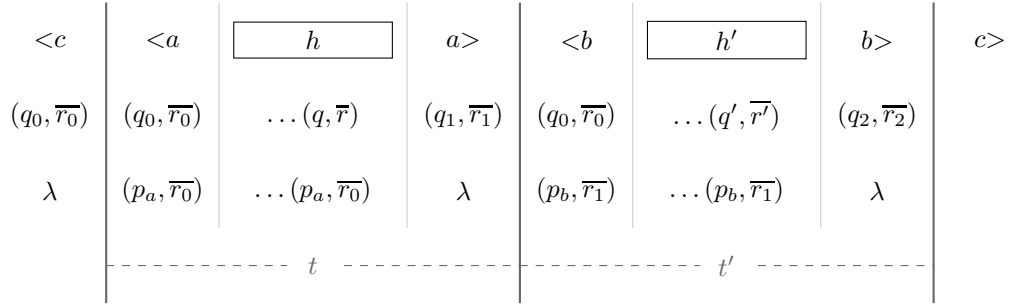
\begin{figure}[h]
			\begin{center}
				\begin{tikzpicture}[every node/.append style={minimum width=20}]
				\node (opf) at (0,0) {${<}c$};
				\node (opa)  at (1.5,0) {${<}a$};
				\node[draw, minimum width=60] (ina) at (3.5,0) {$h$};
				\node (cla) at (5.5,0) {$a{>}$};
				\node (opb) at (7,0) {${<}b$};
				\node[draw, minimum width=60] (inb) at (9,0) {$h'$};
				\node (clb) at (11,0) {$b{>}$};
				\node (clf) at (12.5,0) {$c{>}$};
				
				\node (confopf) [below of= opf] {$(q_0,\ol{r_0})$};
				\node (stackopf) [below of= confopf] {$\lambda$};
				
				\node (confopa) [below of= opa] {$(q_0,\ol{r_0})$};
				\node (stackopa) [below of= confopa] {$(p_a,\ol{r_0})$};
				\node (confinat) [below of= ina] {\dots $(q,\ol{r})$};
				\node (stackinat) [below of= confinat] {\dots $(p_a,\ol{r_0})$};
				\node (confcla) [below of= cla] {$(q_1,\ol{r_1})$};
				\node (stackcla) [below of= confcla] {$\lambda$};
				
				\node (confopb) [below of= opb] {$(q_0,\ol{r_0})$};
				\node (stackopb) [below of= confopb] {$(p_b,\ol{r_1})$};
				\node (confinbt) [below of= inb] {\dots $(q',\ol{r'})$};
				\node (stackinbt) [below of= confinbt] {\dots $(p_b,\ol{r_1})$};
				\node (confclb) [below of= clb] {$(q_2,\ol{r_2})$};
				\node (stackcla) [below of= confclb] {$\lambda$};
				
				\draw[thick, color=black!60] (0.75,0.5) -- (0.75,-3.5);
				\draw[thick, color=black!60] (6.25,0.5) -- (6.25,-3.5);
				\draw[thick, color=black!60] (11.75,0.5) -- (11.75,-3.5);
				
				\draw[color=black!20] (2.25,0.5) -- (2.25,-2.5);
				\draw[color=black!20] (4.75,0.5) -- (4.75,-2.5);
				\draw[color=black!20] (7.75,0.5) -- (7.75,-2.5);
				\draw[color=black!20] (10.25,0.5) -- (10.25,-2.5);

				\draw[dashed, color=black!60] (0.75,-3) edge node[fill=white, minimum width=20] {$t$} (6.25,-3);
				\draw[dashed, color=black!60] (6.25,-3) edge node[fill=white, minimum width=20] {$t'$} (11.75,-3);
				%
				\end{tikzpicture}
			\end{center}
			\caption{The run of a bottom-up STT on a binary tree $c(t,t')$}\label{figSTTannex}
		\end{figure}
		
		From $(q,\ol{r})$, when we read $a{>}$, we use a transition depending only on $q,p_a$ to get $q_1$
		and apply to $\ol r,\ol r_0$ a polynomial function depending only on $q,p_a$.
		We note that $p_a$ depends only of $a$.
		Similarly, from $(q',\ol{r'})$, when we read $b{>}$, we use a transition depending only on $q',p_b$ to get $q_2$
		and apply to $\ol r,\ol{ r_0}$ a polynomial function depending only on $q',p_b$.
		We note that $p_b$ depends of $b$ and $q_1$.
		
		This means that, if we have prior knowledge of $a,b$ -- the roots of the left and right child of $c$ --
		and $q,q'$ -- the states reached by our STT after reading the left and right child of $c$ --
		we have enough information to find the state reached by our STT after reading $c(t,t')$.
		We can call this state $q_{a,b,q,q'}$.
		From $a,b,q,q'$, we can also deduce the polynomial function that links $\ol{r},\ol{r'}$ to $\ol{r_2}$, the values of the registers after reading $c(t,t')$.
		We call such a function $\phi_{a,b,q,q'}$, and it can be seen as a polynomial function on nested words or on OCF.
		
		To find an \aRA{OCF} that computes this function, we say that an input subtree leads to a state $(q,a)$,
		where $a$ is the label of its root, and $q$ is the state reached by the STT right before reading the final $a{>}$.
		The registers have the same values as the STT's right before reading  the final $a{>}$.
		The transitions of our \aRA{OCF} will then be of form:
		$$c((q,a),(q',b))\rightarrow ((q_{a,b,q,q'},c),\phi_{a,b,q,q'})$$
	\end{proof}
	
	\fi
\end{document}